\documentclass{amsart}
\usepackage{graphicx}
\usepackage{amscd}
\usepackage{amsmath}
\usepackage{amsfonts}
\usepackage{amssymb}
\usepackage{setspace}
\usepackage{enumerate}         
\usepackage{fixme}
\usepackage{color}
\usepackage{url}
\usepackage{amsthm}
\usepackage{bm}
\usepackage{xy}
\usepackage{enumitem}

\theoremstyle{plain}
\newtheorem{theorem}{Theorem}[section]

\newtheorem{lemma}[theorem]{Lemma}

\newtheorem{proposition}[theorem]{Proposition}

\newtheorem{definition}[theorem]{Definition}

\newtheorem{assumption}[theorem]{Assumption}
\theoremstyle{remark}
\newtheorem{remark}[theorem]{Remark}
\newtheorem{example}[theorem]{Example}

\numberwithin{equation}{section}

\newcommand{\ind}{1\!\kern-1pt \mathrm{I}}
\newcommand{\rsto}{]\!\kern-1.8pt ]}
\newcommand{\lsto}{[\!\kern-1.7pt [}

\numberwithin{equation}{section}

\newcommand{\la}{\lambda}

\renewcommand{\rho}{\varrho}

\newcommand\ep{\varepsilon}

\DeclareMathOperator{\Slim}{\mathbb{S}-lim}

\begin{document}
\title[FTAP for large financial markets]{A new perspective on the fundamental theorem of asset pricing for large financial markets}
\begin{abstract}
In the context of large financial markets we formulate the notion of \emph{no asymptotic free lunch with vanishing risk} (NAFLVR),
under which we can prove a version of the fundamental theorem of asset pricing (FTAP) in markets with an (even uncountably) infinite number of assets, as it is for instance the case in bond markets. We work in the general setting of admissible portfolio wealth processes as laid down by Y.~Kabanov~\cite{kab:97} under a substantially relaxed concatenation property and adapt the FTAP proof variant obtained in~\cite{CT:14} for the classical small market situation to large financial markets. In the case of countably many assets, our setting includes the large financial market model considered by M.~De Donno et al.~\cite{DGP:05} and its abstract integration theory.

The notion of (NAFLVR) turns out to be an economically meaningful ``no arbitrage'' condition (in particular not involving weak-$*$-closures), and, (NAFLVR) is equivalent to the existence of a separating measure. Furthermore we show -- by means of a counterexample -- that the existence of an equivalent separating measure does not lead to an equivalent $\sigma$-martingale measure, even in a countable large financial market situation.

\end{abstract}
\thanks{The authors gratefully acknowledges the support from ETH-foundation.}

\author{Christa Cuchiero, Irene Klein and Josef Teichmann}
\address{Vienna University, Oskar-Morgenstern-Platz 1, A-1090 Vienna and ETH Z\"urich, R\"amistrasse 101, CH-8092 Z\"urich}
\maketitle

\section{Introduction}

In mathematical finance the classical market model consists of an $\mathbb{R}^d$-valued semimartingale on some filtered probability space which describes the
discounted price process of $d$ financial assets. In the terminology of the present paper, such a model is referred to as \emph{small} financial market, as opposed to
a \emph{large} financial market which corresponds to a sequence of small markets. This concept was introduced by Y.~Kabanov and D.~Kramkov~\cite{KK:94}, who consider a sequence, indexed by $n$, of $\mathbb{R}^{d(n)}$-valued semimartingales on possibly different probability spaces. For such a large financial market model
several notions and characterizations of \emph{asymptotic arbitrage} were developed in~\cite{KK:94,KS:96,K:00,KK:98}.
In particular, a version of the fundamental theorem of asset pricing  was proved in~\cite{K:00}
by adapting the notion of \emph{no free lunch} (NFL) of D.~Kreps~\cite{K:81} to the setting of large financial markets.
In a simplified framework of one fixed probability space it corresponds to (an abstract version of) the Kreps-Yan theorem~\cite{K:81,Y:80} stating the equivalence
of \emph{no asymptotic free lunch} (NAFL) and the existence of an equivalent separating measure.

The goal of the present paper is to prove a similar result, however by replacing the rather difficult to interpret notion of (NAFL)
(involving closures in the weak-$*$-topology on $L^{\infty}$, see Section~\ref{sec:NAFL}
for the precise definition) by an economically more convincing concept, which we call \emph{no asymptotic free lunch with vanishing risk} (NAFLVR) and
which is perfectly in line with the classical \emph{no free lunch with vanishing risk} (NFLVR) condition for small markets first introduced in~\cite{DS:94}.

In order to describe this concept let us consider the example of a market consisting of countably many (discounted) assets
modeled by a sequence of $\mathbb{R}$-valued semimartingales $(S^n_t)_{t \in [0,1]}$
defined on a filtered probability space $(\Omega, \mathcal{F},(\mathcal{F}_t)_{t \in [0,1]}, P)$ for the time interval $[0,1]$.
Similar as in the work by M.~De Donno et al.~\cite{DGP:05}  we define \emph{admissible generalized portfolios} in the large financial market as limits in the
Emery topology of sequences of (uniformly) admissible portfolios
in the small markets, i.e., sequences of portfolios built by trading in finitely many assets such that they are uniformly (in $\omega, t$ and $n$) bounded
from below by some constant $-\lambda$. In the terminology of~\cite{DGP:05, DP:06} these portfolios
correspond to portfolios constructed by using so-called \emph{($\lambda$)-admissible generalized strategies} (compare~\cite[Definition 2.5]{DGP:05})
and formalize the idea that each asset can contribute, possibly with an infinitesimal weight.
More precisely, when considering generalized strategies, one includes portfolios $Z$ which have the property that
for any given $\varepsilon >0$, there is a portfolio
in a small market of which all increments have a distance less than $\varepsilon$ to the increments of $Z$.
This is simply implied by the definition of the Emery metric (see~\eqref{eq:Emetric}), in which
two semimartingales (portfolios) are close if all their increments are close or in more financial terms if all investments (with bounded simple strategies) in the difference of two portfolios are small. The Emery topology can thus be seen as the natural topology on the set of portfolio wealth processes. In particular, it has a clear economic meaning to include Emery limits of portfolios in small markets in no arbitrage requirements in the large financial market, as outlined subsequently.

In perfect analogy to the notion of (NFLVR) in the classical setting of small financial markets, we define (NAFLVR) by
\[
 \overline{C} \cap L^{\infty}_{\geq 0} =\{0\},
\]
where $C$ denotes the convex cone of bounded claims superreplicable (with admissible generalized strategies) at price $0$,
i.e., $C=(K_0-L^0_{\geq 0})\cap L^{\infty}$, and $K_0$ stands for the set of terminal values of admissible generalized portfolios at time $1$.
In words, (NAFLVR) means that there should be no sequence of terminal payoffs of admissible generalized strategies such that their negative parts
tend to $0$ in $L^{\infty}$, while they converge almost surely to a nonnegative random variable
which is strictly positive with positive probability.
Also completely analogous to (NFLVR), (NAFLVR) is equivalent to two conditions, namely \emph{no unbounded profit with bounded risk} (NUPBR)
and \emph{no arbitrage for the large market} (NA).
The (NUPBR) condition means that terminal values of portfolios constructed from $1$-admissible generalized strategies are bounded in probability,
while (NA) requires that almost surely nonnegative terminal values of admissible generalized portfolios are almost surely equal to zero, i.e., $K_0 \cap L^{0}_{\geq 0} =\{0\}$.
With respect to the existing literature on large financial markets, (NUPBR) is equivalent to
\emph{no asymptotic arbitrage of the first kind} (NAA1), which describes the impossibility of getting arbitrarily rich with positive probability by taking an arbitrarily small (vanishing) risk
(see, e.g.,~\cite[Definition 1]{KK:94,KK:98} or~\cite[Definition 1.1]{KS:96}). In particular, the (NAFLVR) condition differs  from previous asymptotic arbitrage notions through the stronger (NA) requirement for the large market.
Indeed, so far (NA) was only required for each small market, but not for the portfolios obtained via generalized strategies. This strengthening allows us to obtain a version
of the fundamental theorem of asset pricing for large financial markets which was so far only available under the above described (NAFL) condition.
In the context of the exemplary  market consisting of countably many assets it reads as follows and is obtained from Theorem~\ref{th:FTAPLFM}:
\begin{theorem}
Let  $(S^n_t)_{t \in [0,1]}$  be a sequence of semimartingales on $(\Omega, \mathcal{F},(\mathcal{F}_t), P)$.
Then there exists an equivalent separating measure $Q \sim P$, i.e., a measure $Q$ such that $E_{Q}[X_1]\leq 0$ for all $X_1 \in K_0$ if and only if
(NAFLVR) is satisfied.
\end{theorem}

Therefore (NAFLVR) can be seen as an economically meaningful ``no arbitrage'' condition
which allows to conclude the existence of a separating measure whose existence was for instance assumed and crucial in the
work on super-replication and utility maximization in large financial markets by M.~De Donno et al.
~\cite{DGP:05}. We do also believe that (NA) is a very reasonable requirement for large financial markets.

In order to allow for a unified treatment of different financial markets, involving for instance a continuum of assets such as in the case of bond markets, we formulate our results in
an extended version of the abstract portfolio wealth process setting introduced in~\cite{kab:97}. The main modification consists in weakening the so-called concatenation property
which in the present situation only holds on a dense set (with respect to the Emery topology) of admissible generalized portfolios in the large financial market.
Despite this weakening most parts of the original proof of the classical fundamental theorem of asset pricing by F.~Delbaen and W.~Schachermayer~\cite{DS:94}
can be transferred to this new setting. We adapt here the proof variant of~\cite{CT:14}, which replaces a series of tricky lemmas in the original proof
via a certain boundedness property in the
Emery topology, called \emph{predictable uniform tightness} (P-UT) property in the literature.

The remainder of the paper is organized as follows: in Section~\ref{sec:setting_kabanov} we introduce a large financial market setting which allows to treat both countably and uncountably
many assets. Section~\ref{sec:DeDonno} is dedicated to explain the precise relation to the large financial
market model of~\cite{DGP:05} as outlined above. The precise definition of (NAFLVR) and the main result, as well as the connection to (NAFL) are stated
in Section~\ref{sec:NAFLVRandFTAP}. Section~\ref{AA1etc} recalls further notions of ``no (asymptotic) arbitrage''
and analyzes their relations, while Section~\ref{sec:wrong}
shows why alternative portfolio wealth process sets in large financial
markets do not lead to the desired result. Section~\ref{sec:sigmamart} provides an example showing
that in large financial markets the existence of an equivalent separating measure does not necessarily yield the existence
of an equivalent $\sigma$-martingale measure. In Section~\ref{proof} we give the proof of the present version of FTAP and Appendix~\ref{sec:app}
concludes with some technical results.

\section{Setting}\label{sec:setting_kabanov}

We modify the setting of Y.~Kabanov introduced in~\cite{kab:97} in order to model large financial markets.
To be precise, let $(\Omega,\mathcal{F},(\mathcal{F}_t)_{t\in[0,1]},P)$ be a filtered probability space satisfying the usual conditions.
Let $\mathbb{S}$ be the space of all $\mathbb{R}$-valued semimartingales $X$ on this filtered probability space defined on $[0,1]$ and starting from zero. The space $\mathbb{S}$ is equipped with the Emery topology defined by the metric
\begin{align}~\label{eq:Emetric}
d_{\mathbb{S}}(X_1,X_2) := \sup_{K \in b \mathcal{E}, \, {\| K\|}_{\infty} \leq 1} E \big[ {|(K \bullet (X_1-X_2))|}^*_1 \wedge 1 \big] \ ,
\end{align}
where $|X|_1^*=\sup_{t\leq 1}|X_t|$,  $b \mathcal{E}$ denotes the set of simple predictable strategies, that is, $K$ is of the form
\[
K=\sum_{i=0}^nK_i1_{]\tau_{i},\tau_{i+1}]} \, ,
\]
with $n\in\mathbb{N}$, stopping times $0 = \tau_0 \leq \tau_1 \leq \dots \leq \tau_n \leq \tau_{n+1}=1$ and $K_i$ are $\mathcal{F}_{\tau_{i}}$-measurable random variables.
The space of semimartingales is a complete metric space with the Emery topology, which follows essentially from the Bichteler-Dellacherie Theorem, see \cite{eme:79}.
Notice that M.~Emery defines the metric via the supremum over all bounded predictable processes (see~\cite{eme:79}) and not only over all simple predictable processes. However, as shown in, e.g., \cite{mem:80}, this leads to equivalent metrics.

Beside convergence in the Emery topology we shall also deal with uniform convergence in probability, which is metrized by
\[
d(X,Y):=E[{|X-Y|}^*_1 \wedge 1 ]  \, ,
\]
and makes the space of c\`adl\`ag processes a complete metric space. Apparently uniform convergence in probability is a weaker topology than the Emery topology.

Let us now formulate a large financial market model.
Let $I\subseteq[0,\infty)$ be a parameter space which can be any subset, countable or uncountable.
Define, for each $n\geq 1$, a family $\mathcal{A}^n$ of subsets of $I$, which contain exactly $n$ elements:
\begin{equation}\label{eq:n-sets}
\mathcal{A}^n=\{\text{some/all subsets }A\subseteq I\text{, such that $|A|=n$}\},
\end{equation}
where $|A|$ denotes the cardinality of the set $A$. Moreover, we assume that if $A^1,A^2\in \bigcup_{n\geq 1}\mathcal{A}^n$, then $A^1\cup A^2\in \bigcup_{n\geq 1}\mathcal{A}^n$.

For each $A\in\bigcup_{n\geq 1}\mathcal{A}^n$ we define the following set of \emph{$1$-admissible
portfolio wealth processes} in the \emph{small financial market $A$}.
\begin{definition}\label{setting}
For each $A\in\bigcup_{n\geq 1}\mathcal{A}^n$ we consider a convex set $\mathcal{X}^A_1\subset \mathbb{S}$ of semimartingales
\begin{itemize}
\item starting at $0$,
\item bounded from below by $-1$,
\item satisfying the following \emph{concatenation property}: for all bounded, predictable strategies $ H,G \geq 0 $, $ X,Y \in \mathcal{X}^A_1 $ with $ HG=0 $ and\\ $ Z = (H \bullet X) + (G \bullet Y) \geq -1 $, it holds that $ Z \in \mathcal{X}^A_1 $.
\item For $A^1, A^2 \in\bigcup_{n\geq 1}\mathcal{A}^n$ with $A^1\subseteq A^2$ we have that $\mathcal{X}^{A^1}_1\subset \mathcal{X}^{A^2}_1$.
\end{itemize}
\end{definition}

Next we define the set
$\mathcal{X}_1^{n}$  of all $1$-admissible portfolio wealth processes with respect to  strategies that include at most $n$ assets (but all possible different choices of $n$ assets). Indeed, for each $n\geq 1$ we consider the following set $\mathcal{X}^n_1\subset \mathbb{S}$ of semimartingales
\begin{equation}\label{eq:X1-n-bonds}
\mathcal{X}_1^{n}=\bigcup_{A\in\mathcal{A}^n}\mathcal{X}_1^{A}.
\end{equation}
Note that the sets $\mathcal{X}_1^{n}$ are neither convex nor satisfy a concatenation property as in Definition~\ref{setting},
because in both cases  $2n$ assets could be involved in the combinations. Therefore the result would be in $\mathcal{X}_1^{2n}$ but not in $\mathcal{X}_1^{n}$.

The following two examples illustrate the main applications that we have in mind. When we deal with multi-dimensional semimartingales and the corresponding multi-dimensional strategies, we use bold letters in order to distinguish them from the one-dimensional analogs.

\begin{example}[{\it Application~1:~Large financial market on one probability space}]\label{LFM}

Similar as in the introduction or as for instance in~\cite{DGP:05}, let us consider a large financial market modeled by
a sequence of semimartingales $(S^n_t)_{t\in[0,1], n \in \mathbb{N}}$, based on a filtered probability space $(\Omega,\mathcal{F},(\mathcal{F}_t)_{t\in[0,1]},P)$.
Our setting covers this model. Indeed, choose $I=\mathbb{N}$ and define
$\mathcal{A}^n=\{A^n\}$ as the singleton containing only the set $A^n=\{1,\dots,n\}$.
The set $\mathcal{X}^n_1=\mathcal{X}^{A^n}_1$ then consists of all stochastic integrals $(\mathbf{H}^n\bullet \mathbf{S}^n)$ where $\mathbf{S}^n=(S^1,\dots, S^n)$
and $\mathbf{H}^n$ is $n$-dimensional, predictable,  $\mathbf{S}^n$-integrable and 1-admissible.
\end{example}

\begin{example}[{\it Application~2:~Markets based on a continuum of assets, such as bond markets}]\label{continuum}

Let us assume that  we are given a continuum of tradeable assets expressed in some num\'eraire. These assets are given as semimartingales
$(S^{\alpha}_t)_{0\leq t\leq 1}$, $\alpha\in I$, based on $(\Omega,\mathcal{F}, (\mathcal{F}_t)_{0\leq t\leq 1},P)$.
One could, for example, choose the parameter set $I=[0,T^*]$, $T^*\leq \infty$,  where $\alpha\in I$ could be thought of as the maturity of a bond. Here, we choose $\mathcal{A}^n$ as the family of all subsets $A\subseteq I$ such that $|A|=n$.

For $A\in\mathcal{A}^n$, where $A=\{\alpha_1,\dots,\alpha_n\}$, for $\alpha_1,\dots,\alpha_n\in I$, the set $\mathcal{X}_1^{A}$ is given by
\begin{equation}
\mathcal{X}_1^{A}=\{(\mathbf{H}^{A}\bullet\mathbf{S}^{A})_{0\leq t\leq T}\text{, $\mathbf{H}^{A}$ 1-admissible} \},
\end{equation}
where $\mathbf{S}^{A}=(S^{\alpha_1},\dots,S^{\alpha_n})$, $\mathbf{H}^{A}$ is an $n$-dimensional, predictable 1-admissible integrand for
$\mathbf{S}^{A}$.
Clearly the sets $\mathcal{X}_1^{A}$ satisfy the requirements of Definition~\ref{setting}.
The sets $\mathcal{X}_1^n$, $n\geq 1$, then express the fact that it is possible to trade in any finite number of assets
(but involving each finite choice of all the uncountably many assets $S^{\alpha},\alpha\in I$).
\end{example}

Continuing the general setting, we now define the sets $\mathcal{X}$ (and $\mathcal{X}_1$ respectively) which will replace the corresponding ones of~\cite{kab:97}
and are referred to as ($1$-) admissible generalized portfolio wealth processes in the large financial market.

\begin{definition}\label{largeX1}
\begin{enumerate}
\item Consider the set $\mathcal{X}_1=\overline{\bigcup_{n\geq 1}\mathcal{X}_1^n}^{\mathbb{S}}$, where $\overline{(\cdot) }^{\mathbb{S}}$ denotes the closure in the Emery-topology.
The elements of $\mathcal{X}_1$ are called  $1$-\emph{admissible generalized portfolio wealth processes} in the large financial market.
\item We denote by $ \mathcal{X}$ the set $\mathcal{X} := \cup_{\lambda >0} \lambda \mathcal{X}_1 $  and call
its elements \emph{admissible generalized portfolio wealth processes} in the large financial market.

\item We denote by $K_0$, respectively $K_0^1$ the evaluations of elements of $\mathcal{X}$, respectively $\mathcal{X}_1$, at terminal time $T=1$.
\end{enumerate}
\end{definition}

We have the following basic assumption:

\begin{assumption}\label{ass:basic_NA}
If $Y \in \mathcal{X} $ and $ Y \geq - \lambda $, for some $ \lambda > 0 $, then there is $ Z \in \mathcal{X}_1 $ such that $ Y = \lambda Z $.
\end{assumption}

In other words: if $ Y $ has a risk greater than or equal to $ - \lambda $, then there are approximating portfolios of the corresponding small financial markets which are all bounded from below by $ - \lambda $. It is not possible that a limiting portfolio with risk greater than or equal to $- \lambda $ of the large financial market can only be approximated by portfolios of the small markets with larger down risks. This is a basic coherence assumption for large financial markets. We could slightly generalize this condition, by replacing $ \mathcal{X}_1 $ by $ \mathcal{X}_{\lambda_0}$, where $\lambda_0$ is a \emph{uniform} bound greater than or equal to $1$ for the approximating portfolios, but we do not see a point for doing so.

\begin{remark}
In the realm of Example~\eqref{continuum}, admissible generalized portfolio wealth processes also include portfolios $Z$ in which the composition of bonds in the portfolio can be changed over time, as long as this change is not too ``wild''. For instance, this is the case if  $Z$ is a semimartingale (corresponding to such a portfolio) and if there
exists a sequence of stopping times $(\tau_k)_{k \in \mathbb{N}}$ converging a.s.~to $1$ such that, for each $k \in \mathbb{N}$, $Z^{\tau_k}$ can be written as
\[
Z_1^{\tau_k}=(\mathbf{H}^{A_k}\bullet\mathbf{S}^{A_k})_{1 \wedge \tau_k}
\]
for some $\lambda$-admissible strategy $\mathbf{H}^{A_k}$  which invests in the bonds corresponding to the index set $A_k$ up to time $\tau_k$.

For example, $Z$ could be built by a strategy that always invests in a finite number of assets, but at certain stopping times the composition of the portfolio can be changed to include a new finite number of assets.
\end{remark}

\begin{lemma}[Concatenation property]\label{lem:concatenation}
The Emery-dense subset $\bigcup_{n\geq 1}\mathcal{X}_1^n$ of $\mathcal{X}_1$ satisfies the following concatenation property: let
$X,Y \in \bigcup_{n\geq 1}\mathcal{X}^n_1$, then
there exists $A\in \bigcup_{n\geq 1}\mathcal{A}^n$ such that $X, Y\in
\mathcal{X}^A_1$. For all
bounded, predictable strategies $ H,G \geq 0 $, with $ HG=0 $ and $ Z = (H \bullet X) + (G \bullet Y) \geq -1 $, it holds that $ Z \in
\mathcal{X}^A_1$ and hence
$Z\in \bigcup_{n\geq 1}\mathcal{X}^n_1$.
\end{lemma}

\begin{proof}
As $X,Y \in \bigcup_{n\geq 1}\mathcal{X}^n_1=
\bigcup_{n\geq 1}\bigcup_{A\in\mathcal{A}^n}\mathcal{X}^A_1$, there is $A^1,A^2\subseteq I$ with $|A^1|=n$ and $|A^2|=m$ such that $X\in\mathcal{X}^{A^1}_1$ and $Y\in\mathcal{X}^{A^2}_1$. Let $A=A^1\cup A^2$, then $A\in
\bigcup_{n\geq 1}\mathcal{A}^n$ and
$X, Y\in \mathcal{X}_1^{A}$, where $A$ has at most $n+m$ elements. Hence, by the concatenation property of the set  $\mathcal{X}_1^{A}$, see Definition~\ref{setting}, we get for $H$ and $G$ as in the statement of the lemma that,
for $Z = (H \bullet X) + (G \bullet Y)$,
$$ Z \in\mathcal{X}_1^{A}\subseteq \bigcup_{n\geq 1}\mathcal{X}^n_1. $$
\end{proof}

\begin{remark}\label{rem:concatenation}
Notice that we do have the concatenation property only on a dense subset of the set of all admissible portfolios, i.e.~we cannot apply the reasoning of \cite{kab:97} or \cite{CT:14} directly. There is no generic way to extend the concatenation property in this infinite dimensional setting to the Emery closure.  Mathematically speaking it is the goal of this work to show that this is still sufficient for all conclusions.
\end{remark}

\subsection{Relation to the literature on generalized trading strategies in large financial markets}\label{sec:DeDonno}

We focus here on the large financial market situation considered by M.~De Donno et al.~in~\cite{DGP:05}, which lies in the realm of Example~\ref{LFM}.
As already outlined there, the set $\mathcal{X}_1^n$ corresponding to $1$-admissible
portfolio wealth processes in the small financial market $n$ is here given by the following set of stochastic integrals with respect to the $n$-dimensional semimartingale
$\mathbf{S}^n=(S^1, \ldots, S^n)$
\begin{align}\label{eq:dedonnoXn1}
\mathcal{X}^n_1=\{(\mathbf{H} \bullet \mathbf{S}^n) \, |\, \mathbf{H} \in \mathcal{H}^n_1 \}
\end{align}
with
\[
\mathcal{H}^n_{\lambda}=\{ \mathbf{H} \, | \, \mathbb{R}^n\textrm{-valued, predictable,  $\mathbf{S}^n$-integrable and $\lambda$-admissible process}\}.
\]
As usual $\lambda$-admissibility means $(\mathbf{H} \bullet \mathbf{S}^n)_t \geq -\lambda$ for all $t \in [0,1]$.

In~\cite{DGP:05} the passage from a sequence of small financial markets to a large financial market model is achieved via a generalized stochastic integration theory
with respect to a sequence of semimartingales. As already mentioned in the introduction, the corresponding integrands are called \emph{generalized strategies}, a notion, which formalizes the idea of a portfolio in which each asset can contribute, possibly with an infinitesimal weight.
As we shall see in  Proposition~\ref{prop:deDonno}, the set $\mathcal{X}_1$ of Definition~\ref{largeX1} corresponds exactly to
stochastic integrals with \emph{$1$-admissible generalized strategies} defined below and initially considered in~\cite{DP:06}.
In the following definition $ \Slim$ denotes the limit in the Emery topology.

\begin{definition}
\begin{enumerate}
\item For each $n \in \mathbb{N}$, let $\mathbf{H}^n$ be an $\mathbb{R}^n$-valued, predictable, $\mathbf{S}^n$-integrable process. A sequence $(\mathbf{H}^n)_{n\in \mathbb{N}}$ is called \emph{generalized strategy} if $(\mathbf{H}^n \bullet \mathbf{S}^n)$ converges in the Emery topology to a semimartingale
\[
Z:=\Slim(\mathbf{H}^n \bullet \mathbf{S}^n),
\]
which is called generalized stochastic integral.
\item Let $\lambda >0$. A generalized strategy $(\mathbf{H}^n)$ is called $\lambda$-admissible if each element of the approximating sequence $\mathbf{H}^n$ is $\lambda$-admissible.
\end{enumerate}
\end{definition}

\begin{proposition}\label{prop:deDonno}
Consider the following set
\[
\mathcal{Z}_1=\{\Slim (\mathbf{H}^n\bullet  \mathbf{S}^n)\,|\, (\mathbf{H}^n) \textrm{ 1-admissible generalized strategy}\}.
\]
Then $\mathcal{Z}_1=\mathcal{X}_1$, where $\mathcal{X}_1$ is given in Definition~\ref{largeX1} with
$\mathcal{X}_1^n$ defined in~\eqref{eq:dedonnoXn1}.
\end{proposition}

\begin{proof}
The inclusion $\mathcal{Z}_1\subseteq \overline{\bigcup_{n\geq 1}\mathcal{X}_1^n}^{\mathbb{S}}=\mathcal{X}_1$ is clear by the definition of 1-admissible generalized strategies.

Conversely let $X \in \mathcal{X}_1$. Then there exists some sequence $X^k \in \bigcup_{n\geq 1}\mathcal{X}_1^n$ such that
$X^k \to X$ in the Emery topology. As each $X^k$ lies in some $\mathcal{X}_1^{n_k}$, it holds that $X^k=\mathbf{H}^{n_k} \bullet \mathbf{S}^{n_k}$ for some
$\mathbf{H}^{n_k} \in \mathcal{H}^{n_k}_1$, 
which proves $\mathcal{X}_1 \subseteq \mathcal{Z}_1 $.
\end{proof}

\begin{remark}\label{rem:motivation}
The one-to-one correspondence between $\mathcal{Z}_1$ and $\mathcal{X}_1$ in the setting of~\cite{DGP:05}, is one motivation for us to define $\mathcal{X}_1$ as $\overline{\bigcup_{n\geq 1}\mathcal{X}_1^n}^{\mathbb{S}}$.
\end{remark}

\section{Main concepts and results}\label{sec:NAFLVRandFTAP}

In order to introduce our notion of absence of asymptotic arbitrage, we need the following convex cones:
\begin{align}\label{eq:coneLFM}
C_0:=K_0 - L^0_{\geq 0}, \quad  C:=C_0\cap L^\infty.
\end{align}

In perfect analogy to \emph{no free lunch with vanishing risk} (NFLVR), we define:

\begin{description}
\item[(NAFLVR)] The set $\mathcal{X}$ is said to satisfy \emph{no asymptotic free lunch with vanishing risk} if
\[
\overline{C}\cap L_{\geq 0}^\infty=\{0\},
\]
where $\overline{C}$ denotes the norm closure in $L^{\infty}$.
\end{description}

\subsection{A fundamental theorem of asset pricing for large financial markets}

Before stating our main result, let us recall the notion of an equivalent separating measure.

\begin{definition}
The set $\mathcal{X}$ satisfies the \emph{(ESM) (equivalent separating measure) property} if there exists an equivalent measure $Q \sim P$ such that $E_{Q}[X_1] \leq 0$ for all $ X \in \mathcal{X}$.
\end{definition}

Note that under the condition
\begin{align}\label{eq:NFL}
\overline{C}^* \cap L_{\geq 0}^\infty=\{0\},\tag{NFL}
\end{align}
where $\overline{C}^*$ denotes the weak-$*$-closure in $L^{\infty}$,
the (ESM) property is a consequence of the Kreps-Yan Theorem~\cite{K:81, Y:80}, which in turn follows from
Hahn-Banach's Theorem. Condition~\ref{eq:NFL} is the classical \emph{no free lunch} (NFL) condition for the abstract set $C$.

It is clear that
$
\text{(NFL)} \Rightarrow \text{(NAFLVR)}\,.
$
The goal is to show the reverse implication,  that is:

\begin{theorem}\label{th:main}
Under (NAFLVR), $ C = \overline{C}^* $, i.e., the cone $C$ is already weak-$*$-closed and (NFL) holds.
\end{theorem}
This result is proved in Section~\ref{proof}, where we adapt the methods used in~\cite{CT:14} to the present setting of large financial markets. Similarly as in the classical setting, this then leads to Theorem~\ref{th:FTAPLFM} below, which is a version of the fundamental theorem of asset pricing for large financial markets.

\begin{theorem}[Fundamental Theorem of Asset Pricing]\label{th:FTAPLFM}
(NAFLVR) $\Leftrightarrow$ (ESM).
\end{theorem}

Note that the corresponding result for a market consisting of finitely many assets was first proved by F.~Delbaen and W.~Schachermayer in~\cite{DS:94}. Y.~Kabanov~\cite{kab:97} found the most abstract setting of admissible portfolio wealth processes as of Definition~\ref{setting} to which the proof of~\cite{DS:94} can be transferred.

\subsection{Connection to \emph{no asymptotic free lunch} (NAFL) and the corresponding FTAP}\label{sec:NAFL}
In this section we discuss the connection of (NAFLVR) to the notion  \emph{no asymptotic free lunch} (NAFL), see~\cite{K:00}.
Here we are again in the situation of Example~\ref{LFM}.
Let
\[
\mathcal{K}_0=\{X_1,\text{ where }X\in \bigcup_{\la>0}\bigcup_{n\geq 1}\la\mathcal{X}_1^n=\bigcup_{n\geq 1} \mathcal{X}^n\},
\]
where $\mathcal{X}_1^n$ is given as in Definition~\ref{setting} (in particular in Example~\ref{LFM}), and $\mathcal{X}^n$ in Remark~\ref{rem:motivation}.
The set $\mathcal{K}_0$ consists of the terminal values of admissible portfolios in all small markets $n$, but without the closure in the Emery-topology. This corresponds to the usual setting of large financial markets under the additional assumption that all processes are based on the same probability space $(\Omega,\mathcal{F},P)$, see, for example~\cite{KK:94, KS:96, K:00}. Define
\begin{equation}\label{eq:coneC} \mathcal{C}=(\mathcal{K}_0-L^0_{\geq 0})\cap L^{\infty}.
\end{equation}
Then (NAFL) reads as follows:
\begin{description}
\item[(NAFL)] The set $\mathcal{C}$ is said to satisfy \emph{no asymptotic free lunch} if
\[
\overline{\mathcal{C}}^*\cap L_{\geq 0}^\infty=\{0\},
\]
where $\overline{\mathcal{C}}^*$ denotes the weak-$*$-closure in $L^{\infty}$.
\end{description}
In order to compare the FTAP of~\cite{K:00} with Theorem~\ref{th:FTAPLFM} above we have to define an appropriate (ESM)-condition.
\begin{definition}
The set $\bigcup_{\la>0} \bigcup_{n\geq 1}\la\mathcal{X}_1^n$ satisfies the (ESM') property if there exists an equivalent measure $Q \sim P$ such that $E_{Q}[X_1] \leq 0$ for all $ X \in \bigcup_{\la>0} \bigcup_{n\geq 1} \la\mathcal{X}_1^n$.
\end{definition}

In~\cite{K:00} the following FTAP for large financial markets was proved in the general setting of a sequence of possibly different probability spaces $\Omega^n$ (where the (NAFL) condition looks more technical, but corresponds to (NAFL) as above if all $\Omega^n$ coincide). In our case when all $\Omega^n$ coincide it can also be deduced from an (abstract version) of the Kreps-Yan~\cite{K:81, Y:80} theorem.

\begin{theorem}\label{th:FTAPNAFL}
(NAFL) $\Leftrightarrow$ (ESM').
\end{theorem}

We will now show  that Theorems~\ref{th:FTAPNAFL} and \ref{th:FTAPLFM} are equivalent
in the sense that we give a precise relation between the sets $C$ and $\mathcal{C}$. To this end we will use the polar sets of $C$ and $\mathcal{C}$. Define
$$\mathcal{M}=\{Q\ll P: \text{$E_Q[f]\leq 0$ for all $f\in C$} \}.$$
It is easy to see that $\mathcal{M}$ is the set of absolutely continuous separating measures for $\mathcal{X}$, that is
$$\mathcal{M}=\{Q\ll P: \text{$E_Q[X_1]\leq 0$ for all $X\in\mathcal{X}$} \}.$$
Moreover we will now see that the set $\mathcal{M}$ is identical with the set of absolutely continuous separating measures in the sense of (ESM').

\begin{lemma}\label{th:compESM}
$Q\in \mathcal{M}$ if and only if $Q$ is an absolutely continuous separating measure for $\bigcup_{\la>0}\bigcup_{n\geq 1}\la\mathcal{X}_1^n$.
\end{lemma}

\begin{proof}
Let $Q$ satisfy the defining condition of (ESM'). In particular, $E_Q[X_1]\leq 0$ for all $X\in
\bigcup_{n\geq 1}\mathcal{X}_1^n$. We have to show that $E_Q[Y_1]\leq 0$ for $Y\in\mathcal{X}_1$. By definition there exists a sequence of processes $X^k\in \bigcup_{n\geq 1}\mathcal{X}_1^n$ such that $X^k\to Y$ in the Emery topology. This implies $X^k_1\to Y_1$ in probability. By assumption $X^k_1\geq -1$, for all $k$, and by (ESM') $E_Q[X^k_1]\leq 0$, for all $k$. So Fatou's lemma implies that $E_Q[Y_1]\leq 0$.

The other direction is clear as
$$\bigcup_{\la>0}\bigcup_{n\geq 1}\la\mathcal{X}_1^n\subseteq
\bigcup_{\la>0}\la\mathcal{X}_1=\mathcal{X}.$$
\end{proof}

Consider the pairing $(L^{\infty}, L^1)$ with the bilinear form $E_P[fg]$ for $f\in L^{\infty}$ and $g\in L^1$.
Let $B$ be any subset of $L^{\infty}$.
Define the polar set $B^{\circ}$ of $B$ as usual
$$B^{\circ}=\{g\in L^1: |E_P[fg]|\leq 1, \text{ for all }f\in B\}.$$
If $B$ is a convex cone (as is the case for $C$ as well as $\mathcal{C}$) the polar is given by
$$B^{\circ}=\{g\in L^1: E_P[fg]\leq 0, \text{ for all }f\in B\}.$$

We will now give the relation of the sets $C$ and $\mathcal{C}$. It turns out, that under (NAFLVR) the set $C$ is exactly the weak-$*$-closure of the set $\mathcal{C}$. This means that the closure in the Emery-topology in the definition of the set $\mathcal{X}_1$ is equivalent to weak-$*$-closing the set $\mathcal{C}$ (which was given similarly as $C$ but without the Emery-closure). So, in fact, the conditions (NAFL) and (NAFLVR) coincide.

\begin{theorem}\label{th:compC}
Suppose that (NAFLVR) holds. Then $C=\overline{\mathcal{C}}^*$.
\end{theorem}

\begin{proof}
We will show  that
$$C^{\circ}=\mathcal{C}^{\circ}=\bigcup_{\alpha\geq 0}(\alpha \mathcal{M}).$$
Indeed each element $g$ of the right hand side satisfies
$E_P[fg]\leq 0$, for each $f\in C$ or $\in\mathcal{C}$, see Lemma~\ref{th:compESM}. Therefore
$\bigcup_{\alpha\geq 0}(\alpha \mathcal{M})\subseteq C^{\circ}, \mathcal{C}^{\circ}$. On the other hand, let $g\in C^{\circ}$ or $\mathcal{C}^{\circ}$. In both cases, as
$$-L_{\geq 0}^\infty\subseteq C, \mathcal{C},$$
we immediately get that $g\geq 0$. Assume the non-trivial case that $P(g>0)>0$. Define $Q$ by $\frac{dQ}{dP}=\frac{g}{E_P[g]}$, then this is an absolutely continuous measure such that $E_Q[f]\leq 0$ for all $f\in C$, or $f \in\mathcal{C}$, respectively. This yields the above claim.

As by assumption (NAFLVR) holds we know that $C$ is a weak-$*$-closed subset of $L^{\infty}$. Moreover, $C$ is a convex cone.
By the Bipolar Theorem the bipolar $C^{\circ\circ}$ is the weak-$*$-closure of $C$, see~\cite{SW:99}. Hence
$$C^{\circ\circ}=C.$$ On the other hand
$\mathcal{C}$ is a convex cone as well. Therefore, again by the Bipolar Theorem,
$$\mathcal{C}^{\circ\circ}=\overline{\mathcal{C}}^*.$$
As by the above reasoning
$C^{\circ}=\mathcal{C}^{\circ}$, it follows that $C=\overline{\mathcal{C}}^*$.
\end{proof}

\subsection{Connection to bond markets as in Klein, Schmidt, Teichmann~\cite{KST:13}}

In the recent article~\cite{KST:13}, treating the particular setting of default-free bond markets with all maturities up to a finite time horizon $T^*$,  a condition is derived, which allows to conclude (NFL) for markets generated by self-financing investments into finitely many bonds (assuming the bond with maturity $T^*$ as the market's num\'eraire). This is also a large financial market with uncountably many assets. However, the condition only involves a countable, dense sub-(large-financial)-market, i.e.~a subset of bonds
(including the terminal one) whose maturities are dense in $[0,T^*]$: it is assumed that the classical (NAFL) condition holds for this
sub-market, and, more importantly, it is assumed that a uniform version of right continuity of the bond prices (w.r.t. to maturity) holds for the whole market. This allows -- in the locally bounded setting of~\cite{KST:13}
-- to conclude the existence of an equivalent local martingale measure for the whole discounted bond market, which in turn means (NAFLVR), or (NFL), in our large financial market setting. Hence the seemingly abstract condition (NAFLVR) can in many circumstances be derived from countable sub-markets under appropriate continuity assumptions.

\section{NA, NUPBR and NAA1}\label{AA1etc}

This section is devoted to recall the above notions of \emph{no (asymptotic) arbitrage} from the literature
and to analyze their relations, also with a view to the proof of Theorem~\ref{th:FTAPLFM} presented in Section~\ref{proof}, where the following equivalence stated in Proposition~\ref{prop:NAFLVR=NUPBRNA} below is needed:
\[
\textrm{(NAFLVR)} \Leftrightarrow \textrm{(NUPBR)+(NA)}.
\]

We start with \emph{no arbitrage} (NA), which is defined analogously to small markets.
In our context of large financial markets it means that almost surely nonnegative terminal values of admissible \emph{generalized} portfolios have to be almost surely equal to zero, which reads in formulas as follows:

\begin{description}
\item[ (NA) ] The set $\mathcal{X}_1$ is said to satisfy \emph{no arbitrage} if
$$K_0\cap L^0_{\geq0}=\{0\},$$
which is equivalent to $C_0\cap L^0_{\geq0}=\{0\}$ and
\begin{align}\label{eq:NAeq}
C \cap L^{\infty}_{\geq 0} =\{0\}.
\end{align}
\end{description}

It is well known that if, in a small financial market satisfying (NA), the terminal value of a portfolio is bounded from below by a constant, the whole portfolio wealth process is bounded from below by this constant.
The following lemma transfers this property to our setting of large financial markets and will be needed in the proof of Theorem~\ref{th:FTAPLFM}. The proof, however, is a bit more involved than in the small markets case.

\begin{lemma}\label{Y_in_X1}
Let $Y \in \mathcal{X}$. If (NA) holds and $Y_1 \geq -1$, then $Y \in \mathcal{X}_1$.
\end{lemma}

\begin{proof}
We will show that $Y\in\mathcal{X}_1$.
By assumption $Y \in \mathcal{X}=\bigcup_{\la>0}\la \mathcal{X}_1$. With $ \lambda:=\sup_{t \in [0,1]} \|Y^-_t\|_{L^{\infty}} >0$ we obtain $Y=\la Z$ with $Z\in \mathcal{X}_1$ by Assumption \eqref{ass:basic_NA}.  We shall prove by contradiction that $\lambda=1$: assume that $Y\notin\mathcal{X}_1$, or equivalently, $\la>1$.

Then there exists $t\in[0,1)$,  $\alpha>0$ and $1>\ep>0$ with $\la-\ep>1$ such that
\begin{equation}\label{setD}
P(Y_t\leq -(\la-\ep))=\alpha>0.
\end{equation}
Let $\delta=\la-\ep-1>0$.
Set $D=\{Y_t\leq -(\la-\ep)\}$. Define the simple predictable integrand
$H_u=\ind_{D}\ind_{]t,1]}$. Then
$$(H\bullet Y)_1=\ind_D(Y_1-Y_t)\geq \ind_D(-1+\la-\ep)=\delta\ind_D,$$
hence $(H\bullet Y)_1\geq 0$ a.s. and $P((H\bullet Y)_1\geq\delta)=\alpha>0$.

Our next step is to show that $(H\bullet Y)\in\mathcal{X}$, then we get a contradiction to (NA). We will show that $(H\bullet Y)$ is actually in $\mathcal{X}_1$. Indeed, by definition of $\mathcal{X}_1$, there exists a sequence $Z^k\in \bigcup_{n\geq 1}\mathcal{X}^n_1$ such that $Z^k\to Z$ in the Emery-topology. There exist finite subsets $A^k$ of $I$ with $Z^k\in\mathcal{X}_1^{A^k}$.  Moreover, clearly, $Y^k:=\la Z^k$ Emery-converges to $Y=\la Z$. It is easy to see that
$$(H\bullet Y^k)\to (H\bullet Y)$$ in the Emery-topology, as for each $K\in b \mathcal{E}, \, {\| K\|}_{\infty} \leq 1$, we have that
$$(K\bullet((H\bullet Y^k)-(H\bullet Y)))=(KH\bullet (Y^k-Y)),$$
and $KH\in b \mathcal{E}, \, {\| KH\|}_{\infty} \leq 1$.

However, we cannot say that $(H\bullet Y^k)\in \mathcal{X}_1^{A^k}\subseteq \bigcup_{n\geq 1}\mathcal{X}^n_1$ because $(H\bullet Y^k)$ might not be $\geq -1$. Therefore we approximate $H$ by predictable simple integrands $H^k=\ind_{D^k}\ind_{]t,1]}$, where
$$D^k=\{Y^k_t\leq -(\la-\ep)+2^{-k}\}.$$
As $Y^k\to Y$ in the Emery-topology, then, in particular,
$Y^k_t\to Y_t$ in probability and hence $\ind_{D^k}\to\ind_D$ in probability.
Moreover, $H^k\geq 0$ and
$$(H^k\bullet Y^k)_u=\ind_{D^k}(Y^k_u-Y^k_t)\geq \ind_{D^k}(-\la+\la-\ep-2^{-k})\geq -\ep\geq-1,$$
for each $u\geq t$. Hence, by an application of the concatenation property Lemma~\ref{lem:concatenation}, $(H^k\bullet Y^k)\in \mathcal{X}_1^{A^k}\subseteq \bigcup_{n\geq 1}\mathcal{X}^n_1$.

Finally, we show that
$$d_{\mathbb{S}}((H^k\bullet Y^k),(H\bullet Y^k))\to0,$$
for $k\to\infty$. Indeed, let $J^k=H^k-H$, then $J^k\ne 0$ only on the set
$E^k=D\setminus D^k\cup D^k\setminus D$, where $P(E^k)\to0$, for $k\to\infty$.
For each $K\in  b \mathcal{E}, \, {\| K\|}_{\infty} \leq 1$ and $u\geq t$ we have that
\begin{align} \left(K\bullet((H^k\bullet Y^k)-(H\bullet Y^k))\right)_u
&=\left((H^k-H)\bullet (K\bullet Y^k)\right)_u\nonumber\\
&=(\ind_{D^k}-\ind_{D})((K\bullet Y^k)_u-(K\bullet Y^k)_t).\nonumber
\end{align}
Hence, for $g^k_u=(K\bullet Y^k)_u-(K\bullet Y^k)_t$,
\begin{align}
&\sup_{K \in b \mathcal{E}, \, {\| K\|}_{\infty} \leq 1}E\big[|K\bullet((H^k\bullet Y^k)-(H\bullet Y^k))|_1^*\wedge 1\big]\nonumber\\
\leq
&\sup_{K \in b \mathcal{E}, \, {\| K\|}_{\infty} \leq 1}E\big[|\ind_{D^k}-\ind_{D}|(\sup_{u\in[0,1]}|g^k_u|\wedge 1)\big],\nonumber\\
\leq &E\big[|\ind_{D^k}-\ind_{D}|\sup_{K \in b \mathcal{E}, \, {\| K\|}_{\infty} \leq 1}(\sup_{u\in[0,1]}|g^k_u|\wedge 1)\big].\label{a}
\end{align}
 But as
$|\ind_{D^k}-\ind_{D}|=\ind_{E^k}\to 0$ in probability and as everything in (\ref{a}) is bounded by $1$
we get by dominated convergence that the expected value in (\ref{a}) converges to $0$. This implies
$d_{\mathbb{S}}((H^k\bullet Y^k),(H\bullet Y^k))\to0$.

This concludes the proof as we found $(H^k\bullet Y^k)\in \bigcup_{n\geq 1}\mathcal{X}^n_1$ with $(H^k\bullet Y^k)\to(H\bullet Y)$ in the Emery-topology and hence
$(H\bullet Y)\in\mathcal{X}_1$ contradicting (NA).
\end{proof}

Let us now introduce the conditions ((N)AA1) and (NUPBR) which turn out to be equivalent.

\begin{definition}\label{def:AA1}
There exists an \emph{asymptotic arbitrage of the first kind} (AA1) if there exist some $\alpha > 0$ and sequences $\varepsilon_k \to 0$, $c_k \to \infty$ and $X^k$ such that for each $k \in \mathbb{N}$
\begin{enumerate}
\item  $X^k \in  \varepsilon_k \bigcup_{n\geq 1} \mathcal{X}_1^n$,
\item $P[ X^k \geq c_k] \geq \alpha $,
\end{enumerate}
\end{definition}

\begin{description}
\item[(NAA1)] \emph{No asymptotic arbitrage of the first kind} holds if there exists no (AA1).\\
\item[(NUPBR)] The set $\mathcal{X}_1$ is said to satisfy \emph{no unbounded profit with bounded risk} if $K^1_0$ is a bounded subset of $L^0$.
\end{description}

Arbitrage of the first kind, which first appeared under this name in~\cite{I:87}, was introduced in the context of large financial markets by Y.~Kabanov and D.~Kramkov~\cite[Definition 1]{KK:94}.
Later it was taken up by C.~Kadaras and I.~Karatzas~\cite{KK:07} in the context of classical small financial markets. In their setting (NA1) is shown to be equivalent to (the corresponding notion of) NUPBR.
The following proposition establishes the analogous result in our large financial market setting.

\begin{proposition}
(NAA1) $\Leftrightarrow$ (NUPBR)
\end{proposition}

\begin{proof}
We first prove (NUPBR) $\Rightarrow$ (NAA1). Assume by contradiction that there is an (AA1). Then there exists a sequence $X^k$ as in Definition~\ref{def:AA1}. For $Y^k=\frac{1}{\varepsilon_k} X^k \in \bigcup_{n\geq 1} \mathcal{X}_1^n$ we have for $k \in \mathbb{N}$
\[
P\left[Y^k \geq \frac{c_k}{\varepsilon_k}\right] \geq \alpha.
\]
This contradicts the boundedness of the set $K_0^1$ in $L_0$. Conversely, assume that $K_0^1$ is not bounded in $L_0$. Then there exists $X^n \in \mathcal{X}_1$ such that
\[
\lim_{n \to \infty} P[X_1^n \geq n] =\alpha >0.
\]
By definition of $\mathcal{X}_1$ there exists a sequence $\bigcup_{n\geq 1}X^n_1 \ni X^{k,n} \stackrel{k \to \infty}\to  X^n$ in the Emery topology. Hence
$X_1^{k,n} \stackrel{k \to \infty}\to  X_1^n$ in probability. For each $n \in \mathbb{N}$, we can choose $k_n$ such that
\[
\lim_{n \to \infty} P\left[X^{k_n,n}_1 \geq \frac{n}{2}\right]\geq \frac{\alpha}{2}.
\]
Then $Y^n:=\frac{1}{\sqrt{n}}X^{k_n,n}$ yields an (AA1).
\end{proof}

\begin{proposition}\label{prop:NAFLVR=NUPBRNA}
(NA) + (NUPBR) $\Leftrightarrow$ (NA) + (NAA1) $\Leftrightarrow$ (NAFLVR)
\end{proposition}

\begin{proof}
The proof is similar to~\cite[Lemma 2.2]{kab:97}.
For the reader's convenience we provide the corresponding arguments in our setting. We start by showing (NAFLVR) $\Rightarrow$ (NA)+ (NAA1). By~\eqref{eq:NAeq},
(NA) follows trivially. Assume there exists an (AA1). Then there exists a sequence $X^k$ as in Definition~\ref{def:AA1}. By defining $f^k:=X^k_1\wedge 1 \in C$ and proceeding as in ~\cite[Lemma 2.2, a) $\Rightarrow$ b)]{kab:97}, we get a contradiction to (NAFLVR).

Conversely, for showing (NA) + (NAA1) $\Rightarrow$ (NAFLVR), suppose (NAFLVR) fails. Then there exist $f^n \in C$ and $f\geq 0$ such that  $P[f >0]>0$ and $\| f^n-f\|_{L^{\infty}} \to 0$.
Hence, there exists  $X_1^n \in \mathcal{X}$ such that $X^n_1 \geq f_n \geq -\| (f^n)^-\|_{L^{\infty}}=:-\varepsilon_n \to 0$.
By Lemma~\ref{Y_in_X1}, (NA) implies that $X^n \in \varepsilon_n \mathcal{X}_1$ and $\frac{1}{\sqrt{\varepsilon_n}} X^n$ yields an (AA1).
\end{proof}

\section{Why alternative portfolio wealth process sets in large financial markets are not suitable}\label{sec:wrong}

Instead of the set $\mathcal{X}$ introduced in Definition~\ref{largeX1}, we here consider two different possibilities to define  portfolio wealth process sets in large financial markets.
In both cases we do not have the desired existence of an equivalent separating measure, which thus illustrates the importance of taking the Emery closure of $1$-admissible portfolios as it is done in Definition~\ref{largeX1}.

\subsection{Importance of \emph{uniformly} admissible approximating portfolio wealth processes}\label{sec:differenceclosure}
Take as set of portfolio wealth processes
\[
\overline{\mathcal{X}}^{\mathbb{S}}=\overline{\bigcup_{\la>0}\bigcup_{n\geq 1}\la\mathcal{X}_1^n}^{\mathbb{S}}.
\]
Then it is \emph{not} possible to conclude that a separating measure for $\mathcal{X}$ is also one for $\overline{\mathcal{X}}^{\mathbb{S}}$.
Indeed, there is a substantial difference in taking the Emery closure of the $1$-admissible portfolios and taking the union over all $\lambda$, as we do in Definition~\ref{largeX1}, and taking  the Emery closure of $\bigcup_{\la>0}\bigcup_{n\geq 1}\la\mathcal{X}_1^n$.
Indeed, an example from~\cite[Example 2]{DP:06} shows that
\[\overline{\bigcup_{\la>0}\bigcup_{n\geq 1}\la\mathcal{X}_1^n}^{\mathbb{S}}=\overline{\mathcal{X}}^{\mathbb{S}} \supset \mathcal{X}=\bigcup_{\la>0} \lambda \overline{\bigcup_{n\geq 1} \mathcal{X}_1^n}^{\mathbb{S}}.
\]
There, a sequence of martingales $(S^n)_{n \in \mathbb{N}}$ and a sequence of $\mathbb{R}^n$-valued strategies $\mathbf{H}^n$ is constructed such that for all $n\in \mathbb{N}$
\[
(\mathbf{H}^n \bullet \mathbf{S}^n) \geq -n,
\]
where $\mathbf{S}^n =(S^1, \ldots, S^n)$. By Ansel-Stricker's lemma (see, e.g.,~\cite{DP:07}), this implies that $(\mathbf{H}^n \bullet \mathbf{S}^n)$ is a local martingale for all $n\in \mathbb{N}$. However, this does not hold true any longer for the limit. Indeed as shown in~\cite[Example 2]{DP:06}, the sequence $(\mathbf{H}^n \bullet \mathbf{S}^n)$ converges in the Emery topology to the increasing process $A_t=t$.

If we now define $\mathcal{X}^n_1$ with respect to the above sequence of  martingales $(S^n)_{n \in \mathbb{N}}$ similar as in Section~\ref{sec:DeDonno}, $(\mathbf{H}^n \bullet \mathbf{S}^n) \in n\mathcal{X}^n_1$ for every $n \in \mathbb{N}$ and thus clearly also in $\bigcup_{\la>0}\bigcup_{n\geq 1}\la\mathcal{X}_1^n$. The limit process $A_t=t$ thus lies in $\overline{\mathcal{X}}^{\mathbb{S}}$. But there is certainly no measure $Q \sim P$ such that $E_{Q}[A_1]=E_{Q}[1]\leq 0$. The set $\mathcal{X}$ in contrast satisfies the (ESM) property. Indeed, in this case only \emph{uniformly} admissible portfolio wealth processes in the small markets are taken into account, i.e., only those strategies for which we have, for all $n \in \mathbb{N}$, $(\mathbf{H}^n \bullet \mathbf{S}^n)\geq -\lambda$ for some $\lambda\geq 0$. By Fatou's lemma this then implies that the corresponding Emery limits are supermartingales, whence the (ESM) property (and thus also (NAFLVR)) holds for $\mathcal{X}$.

In particular, this example shows that the set $\overline{\mathcal{X}}^{\mathbb{S}}$ is in general too large and leads to asymptotic arbitrages which one should exclude from the outset. On the other hand just taking $\bigcup_{\la>0}\bigcup_{n\geq 1}\la\mathcal{X}_1^n$  without any closure will be too small, as shown in the next section.

\subsection{No naive ``no free lunch with vanishing risk''}

Consider similarly to Section~\ref{sec:NAFL} the following set of portfolio processes
\[
\bigcup_{\la>0}\bigcup_{n\geq 1}\la\mathcal{X}_1^n
\]
and the analogous notion of (NAFLVR) for this set, i.e.,
\[
\overline {\mathcal{C}} \cap L_{\geq 0}^{\infty} ,
\]
where $\mathcal{C}$ is defined in~\eqref{eq:coneC} and $\overline {\mathcal{C}}$ denotes the norm closure in $L^{\infty}$.
Moreover denote the evaluations of elements of $\bigcup_{\la>0}\la\mathcal{X}_1^n$ at terminal time $T=1$ by $K_{0,n}$, i.e., $K_{0,n}$
stands for the terminal values of portfolios in the small financial market $n$.
We say (NA$_{\textrm{small}}$) holds if every small financial market $n$ satisfies no arbitrage, i.e., if $K_{0,n} \cap L_{\geq 0}^0 =\{0\}$ for all $n \in \mathbb{N}$. Analogously to Proposition~\ref{prop:NAFLVR=NUPBRNA}, we then obtain the following characterization.

\begin{proposition}
\[
\overline {\mathcal{C}} \cap L_{\geq 0}^{\infty}=\{0\} \Leftrightarrow \textrm{(NAA1)} + \textrm{(NA$_{\textrm{small}}$)},
\]
\end{proposition}

It is possible to construct examples for which (NAA1) and (NA$_{\textrm{small}}$) is satisfied, but we do not get the existence of an equivalent separating measure for $\bigcup_{\la>0}\bigcup_{n\geq 1}\la\mathcal{X}_1^n$.
Indeed consider the following sequence of semimartingales as a large financial market model: for $t \in [0,1)$ and $n \in \mathbb{N}$, let $S^n_t=0$ and $S^n_1$ be independent random variables taking two values $-1$ and $1$ with probability $P[S^n_1=-1]=p_n$ and $P[S^n_1=1]=1-p_n$, where $p_n >0$ for all $n \in \mathbb{N}$ but $\lim_{n \to \infty} p_n=0$. For each $n \in \mathbb{N}$, an equivalent separating measure $Q$ needs to satisfy $Q[S^n_1=-1]=\frac{1}{2}$, which, however, does not tend to $0$ and thus implies that $Q$ is not equivalent to $P$. Note that (NA) in the stronger sense as defined at the beginning of Section~\ref{AA1etc} is violated, since the sequence of investments in the $n^{th}$ asset, i.e., $(\mathbf{H}^n \bullet \mathbf{S}^n)_1=S_1^n-S_0^n$  with $\mathbf{H}^n_{t,i}=1_{[0,1]}(t)\delta_{in}$ yields in the limit an element which is almost surely $1$ and thus an (asymptotic) arbitrage opportunity (corresponding to an asymptotic arbitrage of the \emph{second} kind as introduced in~\cite{KK:94}).

This example shows that under $\overline {\mathcal{C}} \cap L_{\geq 0}^{\infty}=\{0\}$,
$\overline {\mathcal{C}} \subsetneq \overline {\mathcal{C}}^*$.
In view of Proposition~\ref{prop:NAFLVR=NUPBRNA} and Section~\ref{sec:NAFL}, the crucial issue
to obtain an FTAP without using weak-$*$-closures is  thus to strengthen the no arbitrage condition from (NA$_{\textrm{small}}$) to (NA) as defined at the beginning of Section~\ref{AA1etc}.

\section{On $\sigma$-martingale measures in large financial markets}\label{sec:sigmamart}

The purpose of this section is to show that - in contrast to classical small financial markets - (NAFLVR) does not imply the existence of a $\sigma$-martingale measure.
Indeed, the following large financial market model provides a counterexample.

Let $(\Omega, \mathcal{F}, P)=([0,1], \mathcal{B}([0,1]), \lambda)$, where $\lambda$ denotes the Lebesgue measure.
For $t \in [0,1)$ and $n \in \mathbb{N}$, let $S^n_t=0$ and define $S^n_1$ as follows:
\begin{align}\label{eq:nosigma}
S_1^n(\omega) =
\left\{
\begin{array}{ll}
   -\frac{1}{\sqrt{\omega}} & \omega \in [0, \varepsilon_n),\\
      \frac{1}{(1-\omega)^{\frac{1}{n+1}}} &  \omega \in [\varepsilon_n, 1],
\end{array}
\right.
\end{align}
where $(\varepsilon_n)$ is a sequence taking values in $(0,1)$.
We start with a lemma that  shows the existence of a particular sequence  $(\varepsilon_n)$
such that $E[S^n_1]=1$ for all  $n \in \mathbb{N}$ and $\varepsilon_n \to 0$ as $n \to \infty$.

\begin{lemma}\label{lem:eps}
Let $(S^n_1)$ be given by \eqref{eq:nosigma}. Then we can choose $(\varepsilon_n)$ such that
\[
E[S^n_1]=1 \quad \textrm{ for all } n \in \mathbb{N}
\]
and $\varepsilon_n \to 0$ as $n \to \infty$.
\end{lemma}

\begin{proof}
Note that $E[S^n_1]=1$
is equivalent to
\begin{align}\label{eq:equality}
1+2 \sqrt{\varepsilon_n}=\frac{n+1}{n}(1-\varepsilon_n)^{\frac{n}{n+1}}.
\end{align}
On $[0,1]$, the function $x \mapsto f_n(x):=\frac{n+1}{n}(1-x)^{\frac{n}{n+1}}$  (corresponding to the right hand side) is decreasing, while
 $x \mapsto 1 +2 \sqrt{x}$ (corresponding to the left hand side) is increasing. Hence for each $n$ we find some $\varepsilon_n$ such that~\eqref{eq:equality} holds. Moreover, by the same argument the sequence $\varepsilon_n$ is decreasing and converges to $0$ since
$\lim_{n \to \infty} f_n(x)=1-x$.
\end{proof}

In the sequel we take $(\varepsilon_n)$ as of Lemma~\ref{lem:eps}. Moreover, we consider as filtration  the one generated by $(S^n_t)_{t \in [0,1], n\in \mathbb{N}}$ such that all predictable strategies reduce to deterministic ones. The terminal value of a  portfolio in the first $n$ assets is thus of form
$X_1=\sum_{k=1}^n c_k (S_1^k-S_0^k)=\sum_{k=1}^n c_k S_1^k$ for some constants $c_k \in \mathbb{R}$.
The following lemma provides a necessary condition which guarantees that $X_1 \geq -1$ (which is  equivalent to $X \in \mathcal{X}^n_1 $)
and shows that $P=\lambda$ is a separating measure.

\begin{lemma}\label{lem:sepmeas}
Let $(S^n_1)$ be given by \eqref{eq:nosigma} with $(\varepsilon_n)$ as of Lemma~\ref{lem:eps}.
For $n \in \mathbb{N}$, suppose that $X_1=\sum_{k=1}^n c_k S_1^k \geq -1$.
Then
\begin{align}\label{eq:negpos}
\sum_{k \in G_+ } c_k \, \leq \sum_{k \in G_-} |c_k|,
\end{align}
where $G_+=\{ j \, |\,  c_j >0\}$ and $G_-=\{ j \, |\,  c_j <0\}$. Furthermore, $P=\lambda$ is a separating measure for $\mathcal{X}$.
\end{lemma}

\begin{proof}
Concerning the first assertion, suppose by contradiction that
\[
 \alpha:=\sum_{k \in G_+} c_k -\sum_{k \in G_-} |c_k| > 0.
\]
 Then on the set $[0, \varepsilon_n)$, we have $X_1(\omega)=\sum_{k=1}^n c_k S_1^k(\omega)=-\frac{\alpha}{\sqrt{\omega}}$,
 which is unbounded from below and thus implies~\eqref{eq:negpos}.
This together with Lemma~\ref{lem:eps} yields
\begin{align*}
E[X_1]&=\sum_{k \in G_+} c_k E[S_1^k]+\sum_{k \in G_-} c_k E[S_1^k]\\
&=\sum_{k \in G_+} c_k -\sum_{k \in G_-} |c_k| \leq 0
\end{align*}
for all $X \in \mathcal{X}^n_1$, by Fatou's lemma thus also for all $X \in \mathcal{X}_1$ and in turn also for $X \in \mathcal{X}$.
\end{proof}

\begin{lemma}\label{lem:nomart}
Let $(S^n_1)$ be given by \eqref{eq:nosigma} with $(\varepsilon_n)$ as of Lemma~\ref{lem:eps}. Then there does not exist an equivalent martingale measure, i.e., there
is no $Q \sim P=\lambda$ such that
$E_{Q}[S^n_1]=0$ for all $n$.
\end{lemma}

\begin{proof}
Suppose  by contradiction that there exists some $Q \sim \lambda$ such that
$E_{Q}[S^n_1]=0$ for all $n$. This implies that
\begin{align}\label{eq:nomart}
\int_0^{\varepsilon_n} \frac{1}{\sqrt{\omega}}dQ(\omega)=\int_{\varepsilon_n}^1 \frac{1}{(1-\omega)^{\frac{1}{n+1}}}dQ(\omega)
\end{align}
for all $n$. Since $Q \sim \lambda$ and $\varepsilon_n \to 0$ we necessarily have $Q([0, \varepsilon_n)) \to 0$.
Moreover, as  $\frac{1}{\sqrt{\omega}}1_{[0, \varepsilon_1)} \in L^1(Q)$, we obtain by dominated convergence that the left hand side of~\eqref{eq:nomart} tends to $0$ as $n$
tends to $\infty$. Consequently the right hand side tends to $0$ as well, implying that $Q([\varepsilon_n,1]) \to 0$,  since
\[
\int_{\varepsilon_n}^1 \frac{1}{(1-\omega)^{\frac{1}{n+1}}}dQ(\omega)\geq Q([\varepsilon_n,1]).
\]
This contradicts the equivalence of $Q$ and $\lambda$ and proves the assertion.
\end{proof}

Combining the above lemmas we obtain the following result:

\begin{proposition}
 Let $(S^n_1)$ be given by \eqref{eq:nosigma} with $(\varepsilon_n)$ as of Lemma~\ref{lem:eps}. Then, there exists an equivalent separating measure (namely $P=\lambda$), but no
 equivalent martingale measure.
\end{proposition}

\begin{proof}
 The proof is a consequence of  Lemma~\ref{lem:sepmeas} and Lemma~\ref{lem:nomart}.
\end{proof}

\begin{remark}
Note that in the situation of the above example the requirement of $S^n$ being a $\sigma$-martingale is equivalent to $S^n$ being a martingale. We have therefore proved that
in general (NAFLVR) does not imply the existence of an equivalent $\sigma$-martingale measure.
\end{remark}

\section{Proof of Theorem~\ref{th:main}}\label{proof}

We adapt the proof strategy of~\cite{CT:14} to the present setting. This means in particular that the first part of the proof goes along the lines of the original proof by F.~Delbaen and W.~Schachermayer~\cite{DS:94} for the classical small market situation.
The first series of conclusions is the following. Recall also Proposition~\ref{prop:NAFLVR=NUPBRNA}, which is needed to split (NAFLVR) into (NA) and (NUPBR).

\begin{enumerate}
\item The convex cone $C$ defined in~\eqref{eq:coneLFM} is closed with respect to the weak-$*$-topology if and only if $C_0$ is Fatou-closed, i.e.~for any sequence $(f_n)$ in $C_0$ uniformly bounded from below and converging almost surely to $ f $ it holds that $ f \in C_0$, see beginning of Section 3 in~\cite{kab:97}, and~\cite[Theorem 2.1]{DS:94} essentially tracing back to A.~Grothendieck.
\item Take now $ -1 \leq f_n \in C_0 $ converging almost surely to $ f $. Then we can find $ f_n \leq g_n = Y^n_1 $ with $Y^n \in \mathcal{X}$.
\item By Lemma~\ref{Y_in_X1} it follows that each $ Y^n \in \mathcal{X}_1$.
\item By (NUPBR) it follows that there are forward-convex combinations $ \widetilde{Y^n} \in \operatorname{conv}(Y^n,Y^{n+1},\ldots) $ such that $ \widetilde{Y^n_1} \to \widetilde{h_0} \geq f $ almost surely.
\item This implies that the set $\widehat{K}^1_0 \cap \{g \in L_0 \, |\, g \geq f\}$, where $\widehat{K}^1_0$ denotes the closure of $K_0^1$ in $L^0$, is non-empty. Since it is also bounded by (NUPBR) and closed, a maximal element $h_0$ exists (see beginning of Section 3 in~\cite{kab:97} or~\cite[Lemma 4.3]{DS:94}). Since $h_0 \in \widehat{K}^1_0$, we can find a sequence of semimartingales $ X^n \in \mathcal{X}_1 $ such that $ X^n_1 \to h_0 $ almost surely and $ h_0 $ is maximal above $ f $ with this property.
\begin{remark}\label{wlog}
Without loss of generality we can take the sequence of semimartingales in the Emery-dense subset $\bigcup_{n\geq 1}\mathcal{X}^n_1$ of $\mathcal{X}_1$. Indeed, by the properties of the Emery-topology, it is clear that we can find a sequence in the dense subset such that still $X^n_1 \to h_0$ a.s.
\end{remark}
\item \label{convergence} The previously constructed ``maximal'' sequence of semimartingales $ X^n \in \bigcup_{n\geq 1}\mathcal{X}^n_1$ converges uniformly in probability, i.e. $ {|X^n - X|}^*_1 \to 0 $ in probability, to some c\`adl\`ag process $X$ (see~\cite[Lemma 3.2]{kab:97} or~\cite[Lemma 4.5]{DS:94}).
\end{enumerate}

As in~\cite{CT:14} it is now the goal to show that the sequence $(X^n)$ constructed in~\ref{convergence} above, which converges uniformly to $X$ in probability, also converges to $X$ in the Emery topology. From this it follows that $h_0=\lim_{n \to \infty} X^n_1=X_1 \in K_0^1$, since \emph{$\mathcal{X}_1$ is closed in the Emery topology}. This in turn implies that $f \in C_0$, which finishes the proof by step (i) above.

\begin{theorem}\label{th:Emeryconv}
Let $\mathcal{X}_1=\overline{\bigcup_{n\geq 1}\mathcal{X}^n_1}^{\mathbb{S}}$
satisfy (NUPBR). Let $ (X^n)_{n \geq 0} \in \bigcup_{n\geq 1}\mathcal{X}^n_1 $ be a sequence of semimartingales, which converges  uniformly in probability to $X$, such that $X_1$ is a maximal element in $\widehat{K_0^1}$, where $\widehat{K_0^1}$ denotes the closure of $K_0^1$ in $L^0$. Then $ X^n \to X $ in the Emery topology.
\end{theorem}

\begin{proof}
Due to Proposition~\ref{prop:NUPBR-PUT}, (NUPBR) implies the (P-UT) property of $(X^n)$. Hence the theorem is a consequence of Proposition~\ref{prop:Memin} and Proposition~\ref{prop:kabanovmax} below.
\end{proof}

For a sequence of semimartingales $(X^n)_{n \geq 0}$ with $X^n_0=0$ and some $C >0$ let us consider the following decomposition
\begin{align}\label{eq:decomp}
X^n= B^{n,C}+M^{n,C}+\check{X}^{n,C},
\end{align}
where  $\check{X}^{n,C}= \sum_{s\leq t} \Delta X^n_s 1_{\{|\Delta X^n_s| > C\}}$,
$ B^{n,C}$ is the predictable finite variation part and  $M^{n,C}$ the local martingale part of the canonical decomposition of the special semimartingale $X^n-\check{X}^{n,C}$.

\begin{proposition}\label{prop:kabanovmax}
Let $\mathcal{X}_1=\overline{\bigcup_{n\geq 1}\mathcal{X}^n_1}^{\mathbb{S}}$ satisfy (NUPBR) and let $ (X^n)_{n \geq 0} \in  \bigcup_{n\geq 1}\mathcal{X}^n_1$ be a sequence of semimartingales, which converges uniformly in probability to $X$ such that $X_1$ is a maximal element in $\widehat{K_0^1}$. (Here, $\widehat{K_0^1}$ denotes the closure of $K_0^1$ in $L^0$.) Consider the semimartingale decompositions of form~\eqref{eq:decomp} for $(X^n)$ and $X$ and
assume that $M^{n,C} \to M^C$ and $\check{X}^{n,C} \to \check{X}^{C}$ in the Emery topology. Then $B^{n,C}\to B^C$ in the Emery topology.
\end{proposition}

\begin{proof}
We adapt the proof of Proposition 5.5 in~\cite{CT:14} in view of the weakened concatenation property as of Remark~\ref{rem:concatenation}. Let $(X^n) \in \bigcup_{n\geq 1}\mathcal{X}_1^n$. Then $X^n \in \mathcal{X}_1^{A^n}$ for finite subsets $A^n\subseteq I$. Define $Y^n:=M^{n,C}+\check{X}^{n,C}$, which converges in the Emery topology by assumption.

Assume by contradiction that $(B^{n,C})$ does not converge in the Emery topology. Then it is not a Cauchy sequence and there exists $i_k, j_k \to \infty$ such that
\[
P\left[\int_0^1 d|B_s^{i_k,C}-B_s^{j_k,C}| > 2\gamma\right]\geq 2 \gamma >0.
\]

Let $B$ be a predictable increasing process dominating all $B^{n,C}$. Set $r^{i_k}:=\frac{dB^{i_k,C}}{dB},\, r^{j_k}:=\frac{dB^{j_k,C}}{dB}$ and $\Gamma_k:=\{r^{i_k}\geq r^{j_k}\}$.

Then we may conclude -- by assuming $i_k \wedge j_k \geq i_{k-1} \vee j_{k-1}$ and possibly interchanging $i_k$ and $j_k$ -- that
\[
P\left[((r^{i_k}-r^{j_k})1_{\Gamma_k} \bullet B)_1 > \gamma\right]\geq  \gamma >0.
\]
Take $\alpha_k \downarrow 0$ and define $\bar{X}^k:=1_{\Gamma_k}\bullet X^{i_k}+1_{(\Gamma_k)^c}\bullet X^{j_k}$. Note that by Lemma~\ref{lem:concatenation} $\bar{X}^k \in \mathcal{X}_1^{A^{i_k}\cup A^{j_k}}$.
Let us define
\[
\sigma_k=\inf\{ t\, |\, (1_{\Gamma_k}\bullet Y^{i_k})_t+(1_{(\Gamma_k)^c}\bullet Y^{j_k})_t < Y_t^{i_k} \vee Y_t^{j_k}-\alpha_k\}.
\]
Note that $\bar{Y}^{k}-Y^{i_k}=1_{(\Gamma_k)^c}\bullet (Y^{j_k}-Y^{i_k})$ and $\bar{Y}^k-Y^{j_k}=1_{\Gamma_k}\bullet (Y^{i_k}-Y^{j_k})$
converges to $0$ in the Emery topology and thus also uniformly in probability. We  may therefore take $(i_k,j_k)$ growing fast enough to ensure $P[\sigma_k <\infty] \to 0$.
Set now $\widetilde{X}^k:=1_{[0,\sigma_k]}\bullet\bar{X}^k$.
By Lemma~\ref{lem:admissibility}, $\widetilde{X}^k\in (1+\alpha_k) \mathcal{X}^{A^{i_k}\cup A^{j_k}}_1\subset (1+\alpha_k) \bigcup_{n\geq 1}\mathcal{X}_1^n$ and we have the following representation
\begin{align*}
\widetilde{X}^k_1&=(1_{\Gamma_k \cap [0, \sigma_k]}\bullet X^{i_k})_1+(1_{(\Gamma_k)^c \cap [0, \sigma_k]}\bullet X^{j_k})_1\\
&=X^{j_k}_{1\wedge \sigma_k}+(1_{\Gamma_k \cap [0, \sigma_k]}\bullet(X^{i_k}-X^{j_k}))_1\\
& =X^{j_k}_{1\wedge \sigma_k}+(1_{\Gamma_k \cap [0, \sigma_k]}\bullet (Y^{i_k}-Y^{j_k}))_1+\xi_k,
\end{align*}
where $\xi_k=(1_{\Gamma_k\cap [0, \sigma_k]}\bullet(B^{i_k,C}-B^{j_k,C}))_1=(1_{\Gamma_k}(r^{i_k}-r^{j_k})\bullet B)_{1\wedge \sigma_k}$.
Applying~\cite[Lemma A]{kab:97} to $\xi_k$, implies that forward convex combination of $\xi_k$ converge to a random variable $\eta\geq 0$ with $\eta \neq 0$. Denoting the maximal element to which $X_1^{j_k}$ converges by $h_0 \in \widehat{K}_0^1$, it follows -- by the Emery convergence of $(Y^n)$ -- that forward convex combinations of $\widetilde{X}_1^k$ converge to $h_0 +\eta$. Since $\widetilde{X}^k\in (1+\alpha_k)\mathcal{X}_1$, this yields a contradiction to the maximality of $h_0 \in \widehat{K}_0^1$.
\end{proof}

\begin{lemma}\label{lem:admissibility}
Let $X^k\in \mathcal{X}_1^{A^k}, X^l \in \mathcal{X}^{A^l}_1$ (hence $X^k, X^l \in \mathcal{X}^{A^{k}\cup A^{l}}$) and $\alpha >0$. Consider the decomposition~\eqref{eq:decomp} for $X^k$ and $X^l$, i.e.,
\[
X^i=B^{i,C}+M^{i,C}+\check{X}^{i,C}=:B^{i,C}+Y^i, \quad  i=k,l.
\]
Let $B$ be a predictable increasing process dominating $B^{i,C}$. Set $r^i=\frac{dB^{i,C}}{dB}$ for $i=k,l$ and
\[
\sigma:=\inf\{t \, |\, (1_{\Gamma}\bullet Y^k)_t+(1_{\Gamma^c} \bullet Y^l)_t < Y^k_t \vee Y^l_t-\alpha\},
\]
where $\Gamma=\{r^k \geq r^l\}$. Then
\[
\widetilde{X}=1_{\Gamma \cap [0, \sigma]}\bullet X^k+1_{\Gamma^c \cap [0, \sigma]}\bullet X^l \in (1+\alpha)\mathcal{X}^{A^{k}\cup A^{l}}_1 \subseteq(1+\alpha) \bigcup_{n\geq 1}\mathcal{X}_1^n\subseteq (1+\alpha)\mathcal{X}_1
\]
\end{lemma}

\begin{proof}
Notice that $1_{\Gamma} \bullet B^{k,C}+1_{\Gamma^c} \bullet B^{l,C} \geq B^{k,C} \vee B^{l,C}$. Thus on $[0,\sigma[$
\[
\widetilde{X}\geq B^{k,C} \vee B^{l,C}+Y^k \vee Y^l-\alpha\geq (B^{k,C}+Y^k)\vee(B^{l,C}+Y^l)-\alpha=X^k \vee X^l -\alpha.
\]
At time $\sigma$ the jump of $\Delta\widetilde{X}$ is given by $\Delta\widetilde{X}_{\sigma}=1_{\Gamma }\Delta X_{\sigma}^k+1_{\Gamma^c }\Delta X_{\sigma}^l$ and hence $\widetilde{X}_{\sigma}\geq -1-\alpha$, because the left limit is at least $X^k \vee X^l -\alpha$.
\end{proof}

\appendix

\section{Some technical results}\label{sec:app}

\begin{definition}
A positive c\`adl\`ag adapted process $ D $ is called supermartingale deflator for $ 1 + \mathcal{X}_1 $ if $D_0 \leq 1$ and $ D(1+X) $ is a supermartigale for all $ X \in \mathcal{X}_1$.
\end{definition}

\begin{theorem}\label{th:supermartingaledeflator}
Let $ \mathcal{X}_1  $ be a set of semimartingales satisfying Definition \ref{largeX1} and (NUPBR), then there exists a supermartingale deflator $ D $ for $ 1+\mathcal{X}_1$.
\end{theorem}

\begin{proof}
As in the proof of Theorem 4.5 in~\cite{CT:14} the assertion follows from~\cite[Lemma 2.3]{kar-closure:13}. Indeed, following closely the proof of Theorem 4.5 in~\cite{CT:14}, and using the variant of the concatenation property of Lemma~\ref{lem:concatenation}, it is easy to see that \emph{fork convexity} holds for $\bigcup_{n\geq 1}\mathcal{X}_1^n$. Therefore, there exists a supermartingale deflator $D$ for $\bigcup_{n\geq 1}\mathcal{X}_1^n$. By Fatou's lemma $D$ is a supermartingale deflator for the Emery closure $\mathcal{X}_1$ as well.
\end{proof}

The following proposition corresponds to~\cite[Proposition 4.10]{CT:14} and its proof is completely analogous in the present setting.

\begin{proposition}\label{prop:NUPBR-PUT}
Let $\mathcal{X}_1$ satisfy (NUPBR) and let $ (X^n)_{n \geq 0} \in \mathcal{X}_1 $ be \emph{any} sequence of semimartingales. Then $(X^n)$ satisfies the (P-UT) property.
\end{proposition}

The following proposition is a reformulation of~\cite[Proposition 1.10]{MS:91} and corresponds to~\cite[Proposition 5.2]{CT:14}.

\begin{proposition}\label{prop:Memin}
Let $ (X^n)_{n \geq 0} $ be a sequence of semimartingales with $X^n_0=0$, which converges uniformly in probability to $X$.  Assume furthermore the (P-UT) property for this sequence and consider decompositions of form~\eqref{eq:decomp} for $(X^n)$ and $X$. Then there exists some $C >0$ such that $M^{n,C} \to M^C$ and $\check{X}^{n,C} \to \check{X}^{C}$ in the Emery topology and $B^{n,C}\to B^C$  uniformly in probability.
\end{proposition}


\end{document}